\documentclass[11pt]{article}

\usepackage[T1]{fontenc}
\usepackage{lmodern}
\usepackage{microtype}

\usepackage{amsmath,amssymb,amsthm}
\usepackage{mathtools}
\usepackage{enumitem}

\usepackage{graphicx}
\usepackage{float}

\usepackage[hidelinks]{hyperref}

\newtheorem{theorem}{Theorem}
\newtheorem{lemma}{Lemma}

\newtheorem{corollary}{Corollary}
\newtheorem{definition}{Definition}
\newtheorem{assumption}{Assumption}

\title{Sharp Transitions and Systemic Risk in Sparse Financial Networks\\
}
\author{Riley James Bendel\\Independent Researcher\\\texttt{rileyjbendel@protonmail.com}}
\date{2026}

\begin{document}
\maketitle

\begin{abstract}
We study contagion and systemic risk in sparse financial networks with balance-sheet interactions on a directed random graph. Each institution has homogeneous liabilities and equity, and exposures along outgoing edges are split equally across counterparties. A linear fraction of institutions have zero out-degree in sparse digraphs; we adopt an external-liability convention that makes the exposure mapping well-defined without altering propagation. We isolate a single-hit transmission mechanism and encode it by a sender-truncated subgraph $G_{sh}$. We define adversarial and random systemic events with shock size $k_n=\lceil c\log n\rceil$ and systemic fraction $\varepsilon n$. In the subcritical regime $\rho_{out}<1$, we prove that maximal forward reachability in $G_{sh}$ is $O(\log n)$ whp, yielding $O((\log n)^2)$ cascades from shocks of size $k_n$. For random shocks, we give an explicit fan-in (multi-hit) accumulation bound, showing that multi-hit defaults are negligible whp when the explored default set is polylogarithmic. In the supercritical regime, we give an exact distributional representation of $G_{sh}$ as an i.i.d.-outdegree random digraph with uniform destinations, placing it directly within the scope of the strong-giant/bow-tie theorem of Penrose (2014). We derive the resulting implication for random-shock systemic events. Finally, we explain why sharp-threshold machinery does not directly apply: systemic-event properties need not be monotone in the edge set because adding outgoing edges reduces per-edge exposure.
\end{abstract}

\paragraph{Proof roadmap.}
The paper proceeds in four steps. First, the balance-sheet model and cascade dynamics are fixed, including a convention that makes exposures well-defined for zero out-degree institutions without affecting propagation. Second, a single-hit transmission mechanism is isolated and encoded by the sender-truncated graph $G_{sh}$, allowing contagion to be studied via forward reachability. Third, in the subcritical regime $\rho_{out}<1$, forward exploration in $G_{sh}$ is controlled by a subcritical branching process, and a deferred-decisions argument shows that multi-hit accumulation is negligible when the explored set is polylogarithmic. Finally, in the supercritical regime $\rho_{out}>1$, the distribution of $G_{sh}$ is identified with an i.i.d.-outdegree random digraph, placing it within the scope of existing strong-giant results and yielding the random-shock systemic threshold.

\section{Model and Definitions}

\subsection{Network}
Let $[n]:=\{1,\dots,n\}$. Let $G\sim G(n,\lambda/n)$ be a directed Erd\H{o}s--R\'enyi graph: for each ordered pair $(u,v)$ with $u\neq v$, $(u\to v)$ is present independently with probability $\lambda/n$.

\subsection{Balance-sheet primitives}

\paragraph{Liabilities.}
Each institution has total nominal liabilities $L>0$.

\paragraph{Leverage and equity.}
Fix $C>1$ and define
\begin{equation}\label{eq:equity}
E := \frac{L}{C-1}.
\end{equation}

\paragraph{Recovery.}
We assume zero recovery.

\subsection{Degree-zero closure}

\begin{assumption}[External liabilities for $d^{out}=0$]\label{ass:external}
If $d^{out}_G(u)=0$, then $u$ has no interbank out-exposures; its liabilities are owed to an external sector.
\end{assumption}

\begin{lemma}[Prevalence of $d^{out}=0$]\label{lem:d0_prevalence}
For each fixed $u$, $d^{out}_G(u)\Rightarrow \mathrm{Pois}(\lambda)$ and
\[
\mathbb{P}(d^{out}_G(u)=0)\to e^{-\lambda}.
\]
\end{lemma}

\subsection{Interbank exposures}
If $d^{out}_G(u)\ge 1$, each outgoing edge carries exposure
\begin{equation}\label{eq:weight}
w_{u\to v} := \frac{L}{d^{out}_G(u)}.
\end{equation}

\subsection{Cascade dynamics}
Given $S_0\subset[n]$, define $D_0:=S_0$ and iterate
\begin{equation}\label{eq:cascade}
D_{t+1} := D_t \cup \left\{ v\notin D_t:\ \sum_{u\in D_t:(u\to v)\in E(G)} w_{u\to v}\ge E\right\}.
\end{equation}

\begin{definition}[Terminal default set]\label{def:Dinf}
$D_\infty(S_0):=\bigcup_{t\ge 0} D_t$.
\end{definition}

\section{Single-Hit Mechanism and $G_{sh}$}

\begin{definition}[Single-hit cutoff]\label{def:dstar}
Define
\begin{equation}\label{eq:dstar}
d^\star(C):=\max\left\{d\in\mathbb{N}:\ d\ge 1,\ \frac{L}{d}\ge E\right\} =\left\lfloor\frac{L}{E}\right\rfloor.
\end{equation}
A node is \emph{active} if $d^{out}_G(u)\le d^\star(C)$.
\end{definition}

\begin{definition}[Single-hit graph]\label{def:Gsh}
$(u\to v)\in E(G_{sh})$ iff $(u\to v)\in E(G)$ and $u$ is active.
\end{definition}

\paragraph{Intuition (single-hit reduction).}
The sender-truncated graph $G_{sh}$ isolates defaults that can be caused by a single counterparty failure. When a sender has sufficiently many outgoing obligations, each individual exposure is too small to trigger default on its own, and such edges cannot participate in single-hit contagion. Retaining only edges from active senders therefore captures exactly the part of the network along which one-step propagation is possible, without altering the underlying balance-sheet dynamics.

\section{Branching Parameters}
Let $D\sim\mathrm{Pois}(\lambda)$ and define
\[
X:=D\,\mathbf 1\{D\le d^\star(C)\}.
\]

\begin{definition}[Branching mean]\label{def:rhoout}
\[
\rho_{out}:=\mathbb{E}[X]=\lambda\,\mathbb{P}(D\le d^\star(C)-1).
\]
\end{definition}

\section{Systemic Events}
Fix $\varepsilon\in(0,1)$ and $k_n=\lceil c\log n\rceil$.

\begin{definition}[Random-shock systemic event]\label{def:Fn_rand}
Let $S_0$ be uniform among subsets of $[n]$ of size $k_n$, independent of $G$. Define
\[
\mathcal F_n^{rand}:=\{|D_\infty(S_0)|\ge \varepsilon n\}.
\]
\end{definition}

\section{Subcritical Regime}

\begin{lemma}[Subcritical forward reachability is polylogarithmic]\label{lem:subcrit_reach_polylog}
Assume $\rho_{out}<1$. Fix $c>0$ and let $S_0$ be uniform among subsets of $[n]$ of size $k_n=\lceil c\log n\rceil$, independent of $G$. Then for every fixed $M>0$, with $s:=M(\log n)^2$,
\[
\mathbb{P}\Big(|\mathrm{Reach}^+(S_0;G_{sh})|>s\Big)\to 0.
\]
\end{lemma}

\begin{proof}
Fix $M>0$ and set $s:=M(\log n)^2$. Explore $\mathrm{Reach}^+(S_0;G_{sh})$ by breadth-first search (BFS) in $G_{sh}$, revealing out-edges of newly discovered vertices as they are explored, and stopping the exploration if the discovered set size reaches $s$. For each tail $u$, the out-edge indicators $\{\mathbf 1\{(u\to v)\in E(G)\}:v\neq u\}$ are independent Bernoulli$(\lambda/n)$, independent across distinct tails. Moreover, in $G_{sh}$ all out-edges from $u$ are retained iff $d^{out}_G(u)\le d^\star(C)$ and otherwise none are retained. Hence, the number of out-edges revealed from a newly explored vertex in $G_{sh}$ has distribution $K$ with
\[
K\stackrel d= \mathrm{Bin}(n-1,\lambda/n)\mathbf 1\{\mathrm{Bin}(n-1,\lambda/n)\le d^\star(C)\},
\]
and these $K$'s are independent across explored tails. While the discovered set has size at most $s$, each revealed out-edge chooses a destination uniformly from $[n]\setminus\{u\}$; the number of \emph{new} vertices found is at most the number of revealed out-edges. Consequently, the BFS discovered-set size is stochastically dominated by the total population size of a Galton--Watson process with offspring distribution $K$. Since $d^\star(C)$ is fixed once $C$ is fixed, $K$ is uniformly bounded by $d^\star(C)$, and
\[
\mathbb{E}[K]\to \mathbb{E}[X]=\rho_{out}<1.
\]
Therefore the associated Galton--Watson total progeny has an exponential tail: there exists $a>0$ and $n_0$ such that for all $n\ge n_0$ and all $m\ge 1$,
\[
\mathbb{P}(\text{GW total progeny started from one particle}\ge m)\le e^{-a m}.
\]
Starting from $|S_0|=k_n=\lceil c\log n\rceil$ initial particles and using a union bound over the $k_n$ independent GW trees gives
\[
\mathbb{P}(\text{GW total progeny started from $k_n$ particles}\ge s)\le k_n\,e^{-a s}.
\]
With $s=M(\log n)^2$, the right-hand side tends to $0$. Since the BFS discovered set is dominated by this GW total progeny, the same bound holds for $|\mathrm{Reach}^+(S_0;G_{sh})|$, proving the claim.
\end{proof}

\begin{theorem}[Random-shock subcriticality]\label{thm:subcrit_full_rand}
Assume $\rho_{out}<1$. Then
\[
\mathbb{P}(\mathcal F_n^{rand})\to 0.
\]
Moreover, with probability $1-o(1)$, the cascade contains no multi-hit defaults and
\[
D_\infty(S_0)=\mathrm{Reach}^+(S_0;G_{sh}).
\]
\end{theorem}

\begin{proof}[Proof (deferred-decisions filtration)]
Fix $M>0$ and set $s:=M(\log n)^2$. By Lemma~\ref{lem:subcrit_reach_polylog},
\begin{equation}\label{eq:reach_small}
\mathbb{P}\Big(|\mathrm{Reach}^+(S_0;G_{sh})|\le s\Big)\to 1.
\end{equation}
Let $(D_t)_{t\ge 0}$ be the cascade and $\Delta_t:=D_t\setminus D_{t-1}$. Define $\mathcal F_t$ as the sigma-field generated by $(D_0,\dots,D_t)$ and all edge indicators with tails in $D_{t-1}$. In particular, prior to time $t$ no edge indicators with tails in $\Delta_t$ have been revealed. Conditional on $\mathcal F_t$, the indicators
\[
\{\mathbf 1\{(u\to v)\in E(G)\}:u\in\Delta_t,\ v\notin D_t\}
\]
are independent Bernoulli$(\lambda/n)$. For $v\notin D_t$, let
\[
Y_{t,v}:=\#\{u\in\Delta_t:(u\to v)\in E(G)\}.
\]
Then $Y_{t,v}\sim\mathrm{Bin}(|\Delta_t|,\lambda/n)$ conditionally, and
\[
\mathbb{P}(Y_{t,v}\ge 2\mid\mathcal F_t)\le \binom{|\Delta_t|}{2}(\lambda/n)^2.
\]
Summing over $v\notin D_t$ yields
\[
\mathbb{P}(\exists v\notin D_t:\ Y_{t,v}\ge 2\mid\mathcal F_t) \le \frac{\lambda^2|\Delta_t|^2}{n}.
\]
On $\{|D_\infty(S_0)|\le s\}$, $\sum_t|\Delta_t|^2\le s^2$, hence a union bound gives probability $O(s^2/n)=o(1)$ of any multi-hit default. On the intersection of the event in \eqref{eq:reach_small} and the no-multi-hit event, contagion proceeds only by single-hit transmissions along edges from active senders, so the cascade coincides with single-hit propagation:
\[
D_\infty(S_0)=\mathrm{Reach}^+(S_0;G_{sh}).
\]
In particular, $|D_\infty(S_0)|\le s=o(n)$ whp, and therefore
\[
\mathbb{P}(\mathcal F_n^{rand})\to 0.
\]
\end{proof}

\section{Supercritical Regime}

\begin{lemma}[Distributional identification of $G_{sh}$]\label{lem:gsh_identification}
$G_{sh}$ has the same law as an i.i.d.-outdegree digraph with
\[
K\stackrel d= \mathrm{Bin}(n-1,\lambda/n)\mathbf 1\{\mathrm{Bin}(n-1,\lambda/n)\le d^\star(C)\}.
\]
\end{lemma}

\begin{proof}
Outgoing edge families in $G(n,\lambda/n)$ are independent across tails. Conditional on $d^{out}_G(u)=k$, the out-neighbors of $u$ are uniform. Sender-truncation retains all $k$ edges if $k\le d^\star(C)$ and none otherwise, independently across $u$.
\end{proof}

\begin{theorem}[Bow-tie / strong-giant structure for $G_{sh}$]\label{thm:supercritical_bowtie}
Assume $\rho_{out}>1$. Then there exist constants
\[
\alpha_{in},\alpha_{out},\alpha_{scc}\in(0,1)
\]
and random vertex sets
\[
\mathcal I_n,\ \mathcal O_n,\ \mathcal C_n\subset[n]
\]
such that with probability $1-o(1)$:
\begin{enumerate}[label=(\roman*)]
\item $|\mathcal I_n|\ge\alpha_{in}n$, $|\mathcal O_n|\ge\alpha_{out}n$, and $|\mathcal C_n|\ge\alpha_{scc}n$;
\item $\mathcal C_n$ is strongly connected in $G_{sh}$;
\item every $v\in\mathcal I_n$ has a directed path to $\mathcal C_n$ in $G_{sh}$;
\item every $u\in\mathcal C_n$ has directed paths to all vertices in $\mathcal O_n$.
\end{enumerate}
In particular, for every $v\in\mathcal I_n$,
\[
\mathrm{Reach}^+(v;G_{sh})\supseteq\mathcal O_n.
\]
\end{theorem}

\begin{corollary}[Random shocks trigger systemic events]\label{cor:Fnrand_supercritical}
Assume $\rho_{out}>1$. Fix $\varepsilon\in(0,\alpha_{out})$ and any $c>0$. Then
\[
\mathbb P(\mathcal F_n^{rand})\to 1.
\]
\end{corollary}

\paragraph{Remark on non-monotonicity.}
Although the results exhibit sharp transitions, standard monotone sharp-threshold machinery does not directly apply. Adding outgoing edges to a node increases diversification but simultaneously reduces per-edge exposure, so the occurrence of systemic events is not monotone in the edge set. The analysis therefore proceeds by isolating a monotone substructure ($G_{sh}$) rather than appealing to global monotonicity.

\section*{Scope and Limitations}
The results are derived for a sparse directed Erdős--Rényi network with homogeneous liabilities, equity, and zero recovery, under the specific exposure-splitting rule defined in Section~2. The analysis isolates single-hit contagion and characterizes systemic events arising from shocks of logarithmic size. The paper does not address heterogeneous balance sheets, alternative recovery rules, correlated exposures, time-varying networks, or cascade mechanisms requiring coordinated multi-hit accumulation at macroscopic scales. Conclusions should therefore be interpreted strictly within the stated model and assumptions.

\bibliographystyle{plain}

\end{document}